\documentclass[12pt]{article}
\usepackage{amssymb,amsmath,stmaryrd,amsthm}
\usepackage[mathscr]{eucal}
\usepackage{pstricks,pst-node,pst-text,pst-3d}
\usepackage{graphics}
\usepackage{times}

\theoremstyle{definition}

\newtheorem{theorem}{Theorem}
\newtheorem{proposition}{Proposition}
\newtheorem{lemma}{Lemma}
\newtheorem{corollary}{Corollary}

\newtheorem{example}{Example}
\newtheorem{remark}{Remark}

\def \R {\mathbb{R}}
\def \1{{\mathbf{1}}}

\newcommand{\D}{\displaystyle}

\def\R{{\mathbb R}}

\def\core{\mbox{\rm core}}

\def\({\langle}
\def\){\rangle}

\def\cA{{\mathcal A}}
\def\cB{{\mathcal B}}
\def\cC{{\mathcal C}}

\def\cF{{\mathcal F}}

\def\cL{{\mathcal L}}
\def\cM{{\mathcal M}}

\def\cS{{\mathcal S}}

\def\cV{{\mathcal V}}

\def\v0{{\bf 0}}

\def\ov{\overline}

\begin{document}

\title{A Discrete Choquet Integral\\ for Ordered Systems}
\author{Ulrich FAIGLE$^{1}$ and Michel GRABISCH$^{2}$\thanks{Corresponding
    author. Tel (+33) 1-44-07-82-85, Fax
    (+33) 1-44-07-83-01,
    email \texttt{michel.grabisch@univ-paris1.fr}} \\
\normalsize 1. Mathematisches Institut/ZAIK, Universit\"at zu K\"oln \\
\normalsize          Weyertal 80, D-50931 K\"oln, Germany\\
\normalsize          2. Centre d'Economie de la Sorbonne, Universit\'e Paris I\\
\normalsize          106-112, Bd de l'H\^opital, 75013 Paris, France\\
\normalsize         Email: faigle@zpr.uni-koeln.de, michel.grabisch@univ-paris1.fr}
\date{}
\maketitle

\begin{abstract}
A model for a Choquet integral for arbitrary finite set systems is
presented. The model includes in particular the classical model on the system of
all subsets of a finite set. The general model associates canonical non-negative
and positively homogeneous superadditive functionals with generalized belief
functions relative to an ordered system, which are then extended to arbitrary
valuations on the set system. It is shown that the general Choquet integral can
be computed by a simple Monge-type algorithm for so-called intersection systems,
which include as a special case weakly union-closed families. Generalizing
Lov\'asz' classical characterization, we give a characterization of the
superadditivity of the Choquet integral relative to a capacity on a union-closed
system in terms of an appropriate model of supermodularity of such capacities.
\end{abstract}


\section{Introduction}
The Choquet integral \cite{cho53} is a widely used and valuable tool in applied
mathematics, especially in decision theory (see, \emph{e.g.},
\cite{cha94,grla07a,sch98,sch89}). It was characterized by Schmeidler
\cite{sch86} and studied in depth by many authors (see, \emph{e.g.}, Murofushi and
Sugeno \cite{musu91a}). Interestingly, Lov\'asz~\cite{lov83} discovered it
independently in combinatorial optimization
\cite{lov83}, where it has become known as the \emph{Lov\'asz extension} of a set
function. We are particularly interested in Choquet integrals with respect to a
finite universe $N$, the usual environment in applications.

While the classical approach almost always assumes the family of measurable
subsets of $N$ to form an algebra (see also \cite{grla07}, where a ring is
considered), many practical situations (\emph{e.g.}, cooperative games,
multicriteria decision making) require a more general setting with only the
members of a certain subfamily $\cF\subseteq 2^N$ being feasible and no
particular ''nice'' algebraic structure apparent.

In such a general situation, the classical definition of the Choquet integral is
no longer easily utilizable: Many functions become non-measurable  in the
sense that their level sets do not necessarily belong to the family $\cF$.

It is the purpose of the present paper, to extend the notion of a Choquet
integral to arbitrary families $\cF$ of subsets in such a way that functions can
be integrated with respect to general set functions (and capacities being a
particular case). To do so, we consider $\cF$ as an ordered system (whose 
order relation may arise from a particular application model under
consideration).

Our model may be viewed as a discrete analogue of the idea of Riemannian sums in
the usual approach of integration theory.  We consider the approximation of
functions by step functions from below, focussing first on belief functions
(a.k.a. infinitely monotone capacities or positive games) as integration
measures. The key in our construction is the fact the set of non-negative,
positively homogeneous functionals that provide upper approximations to a belief
function have a well-defined unique lower envelope
(Lemma~\ref{l.Choquet-upper-integral}), which yields the Choquet integral in the
classical case by linear extension. In the general model, we extend it to
arbitrary set functions via its decomposition into a difference of belief
functions (Sections~2 and 3). The classical model is a special case of our
approach.

In Section~4, we introduce a heuristic to compute the general Choquet integral,
via a Monge-type (or greedy) algorithm. We prove that the Monge algorithm
computes the integral correctly for all valuations (integration measures) if it
is correct for simple belief functions (a.k.a. unanimity games).  Section~5
studies the important case of the set-theoretic containment order. In particular,
the so-called weakly union-closed families and algebras (where we recover Lehrer's
\cite{leh09} integral) are studied. In Section~6, intersection systems (which are
related to flow networks) are discussed. For these families, it is proved that
the Monge algorithm computes correctly the Choquet integral. Moreover, we study
under which conditions the integral is superadditive and, in particular,
generalize Lov\'asz' characterization of superadditive Choquet integrals to
general (weakly) union-closed structures.

\section{Fundamental notions}
An \emph{ordered system} is a pair $(\cF,\preceq)$, where $\cF$ is a family
of non-empty subsets of some set $N$ with $n:=|N|<\infty$ that covers all
elements of $N$, \emph{i.e.},
$$
 \bigcup_{F\in \cF} F =N,
$$
(partially) ordered by the precedence relation $\preceq$.

\begin{remark} The assumption that $\cF$ be partially ordered does not restrict the generality of our model: Any family $\cF$ can, for example, be \emph{trivially ordered} by setting
$$ F\preceq G \quad:\Longleftrightarrow\quad F=G.
$$
A classical example of order relation is set inclusion. More general orders will
be introduced in what follows.
\end{remark}

We set $m:=|\cF|$ and, for notational convenience, arrange (index) the members
of $\cF=\{F_1,\ldots,F_m\}$ in a monotonically decreasing fashion, \emph{i.e.}, such
that
\begin{equation}\label{eq.index-order}
      F_i\succeq F_j \quad \Longrightarrow \quad i\leq j  \quad(1\leq i,j\leq m).
\end{equation}

\subsection{Valuations and weightings}\label{sec:valuations-weightings}
A \emph{valuation} on $\cF$ is a function $v:\cF\to \R$. The $m$-dimensional
vector space of all valuations on $\cF$ is denoted by $\cV=\cV(\cF)$. Whenever
convenient, we identify $\cV$ with the vector space $\R^{\cF}$, which in turn
can be identified with the $m$-dimensional parameter space $\R^m$ via the index
rule (\ref{eq.index-order}).

\begin{remark}
Setting $\cF_0:=\cF\cup\{\emptyset\}$ and $v(\emptyset):=0$, valuations are
usually called \emph{games} defined on a subfamily of $2^N$. If in addition
  $v$ is non-negative and \emph{isotone} (or \emph{monotone}) w.r.t. $\preceq$
(\emph{i.e.}, $v(F)\leq v(G)$ whenever $F\preceq G$), we call $v$ a
\emph{capacity} or a \emph{fuzzy measure}, refering respectively to the work of
Choquet's~\cite{cho53} and Sugeno \cite{sug74}.
\end{remark}

We define the inner product of any $v,y\in \R^\cF$ as usual via
$$
     \langle v,y\rangle  := \sum_{F\in \cF} y_F v(F).
$$

\medskip
A \emph{weighting} is a function $f:N\to \R$. So the $n$-dimensional space $\R^N$ of all weightings could be identified with $\R^n$ if we fixed a linear arrangement of the elements of $N$. We set
$$
    \langle f,x\rangle  := \sum_{i\in N} f_ix_i \quad\mbox{for all $f,x\in \R^N$.}
$$

It is convenient to identify a subset $A\subseteq N$ with its \emph{incidence function} $\1_A: N\to \{0,1\}$, where
$$
    \1_A(i) = 1 \quad \Longleftrightarrow\quad i\in A,
$$
and to use the notation $x(A) := \langle \1_A,x\rangle $ for any $x\in \R^N$.

\subsection{Belief functions}\label{sec:density-belief}
A \emph{density} is a non-negative valuation $w:\cF\to\R_+$ and gives rise to its associated
cumulative function $\hat{w}:\cF\to \R$ with
$$
\hat{w}(F):= \sum_{F'\preceq F} w(F') \quad\mbox{for all $F\in \cF$.}
$$
We say that the valuation $v\in \cV$ is a \emph{belief function} if $v=\hat{w}$ is the cumulative function of some non-negative $w\in \cV$.

\begin{remark} Our definition generalizes the notion of \emph{belief
    functions} proposed by Shafer~\cite{sha76}, where
  $(\cF_0,\preceq)=(2^N,\subseteq)$. Note that belief functions are normalized ($v(N)=1$) in the classical definition. 
\end{remark}

Note that there is a one-to-one correspondence between densities and
positive functions. To see this, consider the \emph{incidence} matrix
$Z=[z_{ij}]$ of $(\cF,\preceq)$, where
$$
    z_{ij} := \left\{\begin{array}{cl} 1 &\mbox{if $F_i\preceq F_j$}\\
    0 &\mbox{otherwise.}\end{array}\right.
$$
$Z$ is (lower) triangular with diagonal elements $z_{ii} = 1$ and hence invertible. So we have for all $v,w\in \R^{\cF}$ (considered as row vectors):
\begin{equation}\label{eq.Moebius-inversion}
   w = \hat{v}  \quad \Longleftrightarrow \quad v = wZ^{-1}.
\end{equation}
\begin{remark}
The inverse $Z^{-1} =[\mu_{ij}]$ of the incidence matrix $Z$ is called the
\emph{M\"obius matrix} (or more classically the \emph{M\"obius function}) of the
order $(\cF,\preceq)$ and the relationship (\ref{eq.Moebius-inversion}) is known
as \emph{M\"obius inversion}: $v$ is called the \emph{M\"obius inverse} or
\emph{M\"obius transform} of $w$. M\"obius inversion is well known in capacity
theory and decision making. (For a general theory of M\"obius algebra see,
\emph{e.g.}, Rota~\cite{rot64}.)
\end{remark}

\medskip
A \emph{simple (belief) function} is a valuation $\zeta^i:\cF\to\{0,1\}$ with the defining property
$$
     \zeta^i(F) = 1 \quad:\Longleftrightarrow\quad F\succeq F_i \quad(F\in \cF).
$$
So $\zeta^i$ corresponds to the $i$th row of the incidence matrix $Z$ of
$(\cF,\preceq)$, which implies that the set $\{\zeta^1,\ldots,\zeta^m\}$ is a basis for
the valuation space $\cV$,  to which we refer as the \emph{incidence basis}. Whenever
convenient, we will use the notation $\zeta^F$ instead of $\zeta^i$  for the simple
function associated to $F=F_i$. Observe that the \emph{M\"obius relation}
\begin{equation}\label{eq.Moebius-relation}
    \hat{w} = wZ \quad\mbox{for all densities $w:\cF\to\R_+$}
\end{equation}
 exhibits belief functions as precisely the valuations in the simplicial cone $\cV^+$  generated by the simple functions, where
$$
   \cV^+ := \left\{\sum_{i=1}^m \beta_i\zeta^i\mid \beta_1,\ldots,\beta_m\geq 0\right\}.
$$
In particular, simple functions are belief functions in their own right.

\medskip
With any $v=\sum_{i=1}^m \beta_i\zeta^i\in \cV$, we associate the belief functions
$$
v^+:= \sum_{\beta_i\geq 0} \beta_i \zeta^i \quad\mbox{and}\quad v^-:= \sum_{\beta_j\leq 0} (-\beta_j) \zeta^j
$$
and thus obtain the natural representation
\begin{equation}\label{eq.belief-representation}
v = v^+ -v^- \quad\mbox{with $v^+,v^- \in \cV^+$.}
\end{equation}

\begin{remark}
Simple functions generalize so-called \emph{unanimity games}.
\end{remark}

\subsection{The classical Choquet integral}
Assume $(\cF_0,\preceq)=(2^N,\subseteq)$ and let $v:\cF_0\to \R$ be a game. For any non-negative
vector $f\in\R^n_+$, the \emph{(classical) Choquet integral}~\cite{cho53}
w.r.t. $v$ is defined by
\begin{equation}\label{eq:ccho}
\int f\, dv := \int_0^\infty v(\{i\in N\mid f_i\geq \alpha\})d\alpha.
\end{equation}
The definition immediately yields:
\begin{equation}\label{eq:1}
\int \1_A\, dv = v(A).
\end{equation}
The Choquet integral is non-decreasing w.r.t. $v$. Moreover, it is non-decreasing
w.r.t. $f$ if and only if $v$ is a capacity.

\medskip
Letting $F\mapsto \beta_F$ with $\beta_\emptyset:=0$ be the M\"obius inverse of
$v$ relative to $(2^N,\subseteq)$, the following representation of the Choquet
integral is well-known (see, \emph{e.g.}, \cite{chja89}):
\begin{equation}\label{eq:mobcho}
\int f\, dv = \sum_{F\subseteq N} \beta_F\min_{i\in F}f_i.
\end{equation}

Two functions $f,f'\in\mathbb{R}^N$ are \emph{comonotonic} if there are no
$i,j\in N$ such that $f_i>f_j$ and $f'_i<f'_j$ (equivalently, if the combined
level sets $\{i\in N\mid f_i\geq\alpha\}$, $\{i\in N\mid g_i\geq\alpha\}$ form a chain).
A functional $I: \mathbb{R}^n\rightarrow\mathbb{R}$ is \emph{comonotonic
  additive} if $I(f+f') = I(f) + I(f')$ is true for any two comonotonic $f,f'\in \R^n$.

\medskip
The next result is a direct consequence of Theorem 4.2 in \cite{musuma94} and
generalizes Schmeidler's~\cite{sch86} characterization of
the Choquet integral w.r.t. a capacity (where positive homogeneity is
replaced by stipulating that $I$ be non-decreasing w.r.t. the integrand).

\begin{proposition}[Characterization w.r.t. a set function]\label{prop:schm1}
  The functional $I: \mathbb{R}^n\rightarrow\mathbb{R}$ is the Choquet
  integral w.r.t. a set function $v$ on $2^N$ if and only if $I$ is positively
  homogeneous, comonotonic additive, and $I(0)=0$. Then $v$ is uniquely
  determined by (\ref{eq:1}).
\end{proposition}

Note that the functional $f\mapsto \int f\,dv$  is positively homogeneous and hence is concave
if and only if it is superadditive. Important is Lov\'asz'~\cite{lov83} observation (which we will generalize in Section~\ref{sec:union-closed-systems}):

\begin{proposition}\label{prop:mono}
The functional $f\mapsto \int f\,dv$ is concave if and only if $v$ is
supermodular, {\it i.e.}, if $v$ satisfies the inequality 
\begin{equation}\label{eq:supermod}
v(F\cup G) + v(F\cap G) \geq v(F) + v(G) \quad\mbox{for all $F,G\subseteq N$.} 
\end{equation}
\end{proposition}

\medskip
Since belief functions (relative to $(2^N,\subseteq)$) are supermodular, we find:
\begin{itemize}
\item  $f\mapsto \int f\, dv$ is positively homogeneous and superadditive for any $v\in \cV^+$ and extends the set function $v$ ({\it via}  $v(F) = \int 1_F \,dv$ for all $F\subseteq N$).
\end{itemize}

It turns out that a positively homogeneous and superadditive functional with
the extension property may not exist for a general ordered system
$(\cF,\preceq)$. We will show, however, that a well-defined best approximation
to the extension property always exists, which allows us to introduce a general
Choquet integral in analogy with Riemann sums.

\section{Integrals}
We now construct the discrete Choquet integral for an ordered system $(\cF,\preceq)$ in several steps and first consider belief functions.

\subsection{Upper integrals}\label{sec:upper-integrals}
An \emph{upper integral} for the belief function $v\in \cV^+$ is a non-negative,
positively homogeneous and superadditive functional $[v]:\R^N\to \R_+$ that
dominates $v$. In other words, the upper integral $[v]$ has the properties
\begin{itemize}
\item[(i)] $[v](\lambda f) = \lambda[v](f)\geq 0$\; for all scalars $\lambda\geq 0$.
\item[(ii)] $[v](f+g) \geq [v](f) +[v](g)$\; for all $f,g\in \R^N_+$.
\item[(iii)] $[v](\1_F) \geq v(F)$\; for all $F\in \cF$.
\end{itemize}

\medskip
The key observation is that the class of upper integrals of $v\in \cV^+$ possesses a unique lower envelope $v^*$. To see this, we introduce the polyhedron
\begin{equation}\label{eq.core}
   \core(v) := \{x\in \R^N_+\mid x(F)\geq v(F), \;\forall F\in \cF\}.
\end{equation}
\begin{lemma}\label{l.Choquet-upper-integral}
For any $v\in \cV^+$, there is a unique upper integral $v^*$ that provides a
lower bound for all upper integrals $[v]$ in the sense
$$
    v^*(f) \leq [v](f) \quad\mbox{for all $f\in \R^N_+$.}
$$
Moreover, one has
\begin{eqnarray*}
 v^*(f) &=& \min \{\langle f,x\rangle  \mid x\in \R^N_+, x(F)\geq v(F), \;\forall F\in \cF \}\\
 &=&\max \Big\{\langle v,y\rangle \mid y\in \R^{\cF}_+, \sum_{F\in \cF}y_F\1_F\leq f\Big\}.
\end{eqnarray*}
\end{lemma}
\begin{proof}
Associate with the upper integral $[v]$ its \emph{kernel} as the closed convex set
$$
\ker[v] := \{x\in \R^N_+\mid \langle f,x\rangle \geq [v](f) \;\forall f\in \R^N_+ \}.
$$
Since $[v]$ is positively homogeneous and superadditive (properties (i) and (ii)), standard results from convex analysis (see, e.g., Rockafellar~\cite{roc70}) yield
$$
     [v](f) = \min_{x\in \ker[v]} \langle f,x\rangle  \quad\mbox{for all $f\in \R^N_+$.}
$$
By property (iii), we have $\ker[v]\subseteq \core(v)$, which implies
$$
[v](f) \geq  \min \{\langle f,x\rangle  \mid x\in \R^N_+, x(F)\geq v(F) \;\forall F\in \cF \} =: v^*(f).
$$
Linear programming duality yields the representation
\begin{equation}\label{eq.Riemann-sum}
      v^*(f) = \max \Big\{\langle v,y\rangle \mid y\in \R^{\cF}_+, \sum_{F\in \cF}y_F\1_F\leq f\Big\}.
\end{equation}

It is straightforward to verify that $f\mapsto v^*(f)$ is an upper integral for $v$.
\end{proof}

\begin{remark}
The representation (\ref{eq.Riemann-sum}) may be thought of as a Riemann sum
approximation of $v^*(f)$: One approximates $f$ from below by ''step functions''
$\sum_{F\in \cF}y_F\1_F$ and optimizes over their ''content'' $\langle
v,y\rangle$. The same approach has been taken by Lehrer, who
calls it the concave integral \cite{leh09,lete08}, with the difference that
$\cF=2^N$ and that $v$ can be any capacity. 
\end{remark}

\subsection{The Choquet integral}\label{sec:Choquet-integral}
We call the upper integral $v^*$ established in
Lemma~\ref{l.Choquet-upper-integral} the \emph{Choquet integral} of the belief
function $v\in \cV^+$ and henceforth use the notation
$$
\int_{\cF} f \, dv := v^*(f).
$$

\begin{remark} The name of the integral and its notation will be justified later (\emph{cf.} formula (\ref{eq.classical-Choquet}) in Section~\ref{sec.w-u-closed})  as a generalization of the classical Choquet integral, \emph{i.e.}, when
$(\cF,\preceq)=(2^N\setminus\{\emptyset\},\subseteq)$.
\end{remark}

\begin{proposition}\label{p.subadditivity} The functional $v\mapsto \int_\cF f\,dv$ is subadditive on $\cV^+$.
\end{proposition}

\begin{proof}
For any $v,w\in \cV^+$ and $f\in \R_+^N$, we have
\begin{eqnarray*}
\int_{\cF} f \,d(v+w)&=& \min\{ \(f,x\)\mid x\geq 0, x(F)\geq v(F)+w(F) \;\forall F\in \cF\}\\
&\leq&  \min\{ \(f,x\)\mid x\geq 0, x(F)\geq v(F)\;\forall F\in \cF\}\\
&& + \min\{ \(f,x\)\mid x\geq 0, x(F)\geq w(F) \;\forall F\in \cF\}\\
&=& \int_\cF f\, dv + \int_\cF f\,dw.
\end{eqnarray*}
\end{proof}

\medskip
We extend the Choquet integral to arbitrary valuations $v\in \cV$ via
$$
   \int_\cF f \, dv := \int_\cF f \, dv^+ - \int_\cF f \, dv^- \quad\mbox{for all $f\in \R^N$.}
$$
Note that the Choquet integral is positively homogeneous for any valuation.

\subsubsection{Choquet integrals of arbitrary weightings}\label{sec:arbitrary-f}
We call the Choquet integral \emph{strong} if it satisfies
\begin{equation}\label{eq.strong-basis}
      \int_\cF  (f+ \lambda\1_N) \, dv = \int_\cF  f \, dv +\lambda\int \1_N\, d v
\end{equation}
for all $\lambda\geq 0$, $f\in \R_+^N$, $v\in \cV^+$. Given an arbitrary
weighting $f:N\to\R$ bounded from below, we now select some $\lambda\geq 0$ so that
$\ov{f}=f+\lambda\1_N\geq 0$ holds and set
\begin{equation}\label{eq.arbitrary-integral}
\int_\cF  f \, dv := \int_\cF \ov{f} \, dv -  \lambda\int_\cF  \1_N \, dv \quad\mbox{for all $v\in \cV$.}
\end{equation}
In the case of strongness, it is easy to see that (\ref{eq.arbitrary-integral})
is well-defined (\emph{i.e.}, independent of the particular $\lambda$ chosen).

\section{The Monge algorithm}\label{sec:Monge}
We now present a heuristic algorithm for the computation of the Choquet integral
relative to the ordered system $(\cF,\preceq)$, which generalizes the well-known \emph{north-west corner rule} for the solution of assignment problems. As usual, we denote the empty string by $\Box$. Also, we set
$$
\cF(X) := \{F\in \cF\mid F\subseteq  X\} \quad\mbox{for all $X\subseteq N$.}
$$

 \medskip
 Given the non-negative weighting $f\in \R^N_+$, consider the following procedure:

\bigskip
\textsc{Monge Algorithm (MA):}

\medskip
\begin{itemize}
\item[(M0)] Initialize: $X\leftarrow N$, $\cM\leftarrow \emptyset$, $c\leftarrow f$, $y\leftarrow 0$, $\pi\leftarrow \Box$;
\item[(M1)] Let $M=F_i\in \cF(X)$ be the set with minimal index $i$ and choose an element $p\in M$ of minimal weight $c_p=\min_{j\in M}c_j$;
\item[(M2)] Update: $X\leftarrow X\setminus\{p\}$, $\cM\leftarrow\cM\cup\{M\}$,
  $y_{M}\leftarrow c_p$, $c\leftarrow (c-c_p\1_{M})$, $\pi\leftarrow (\pi p)$;
\item[(M3)] If $\cF(X) = \emptyset$, \text{Stop} and Output $(\cM,y,\pi)$. Else goto (M1);
\end{itemize}

It is straightforward to check that in each iteration of (MA) the current vector $y$ is non-negative with the property
$$
\sum_{F\in \cF(X)}y_F \1_F \leq c.
$$
So we find:

\begin{lemma}\label{lem:ma}
The output $\ov{y}$ of the Monge algorithm satisfies for any input $f\in \R_+^N$
$$
         \ov{y}\geq 0 \quad\mbox{and}\quad \sum_{F\in \cF} \ov{y}_F \1_F\leq f.
$$

\end{lemma}

Given any valuation $v\in \cV$, associate with the output $(y,\pi)$ of MA the quantity
$$
     [f](v) := \langle v,y\rangle  = \sum_{F\in \cF} y_F v(F).
$$
Since $(y,\pi)$ does not depend on $v$, it is clear that $v\mapsto [f](v)$ is
a linear functional on $\cV$ (which may depend not only on $f$ and $\cF$ but also on the indexing of $\cF$ and the choice of $p$ in step (M1), however).

\begin{theorem}\label{t.Monge} The following statements are equivalent:
\begin{itemize}
\item[(a)] $[f](\zeta^i) = \int_\cF f \,d\zeta^i$\; for all $i=1,\ldots,m$.
\item[(b)] $[f](v) = \int_\cF f \, dv$\; for all $v\in \cV$.
\end{itemize}
\end{theorem}
\begin{proof}
We have to verify the non-trivial implication $(a) \Rightarrow (b)$. Let us call $\ov{y}$
the output of MA for $f$. From Lemma~\ref{lem:ma}, $\langle
\zeta^i,\ov{y}\rangle =\int f \,d\zeta^i$ means that the $\ov{y}$ is optimal for
the linear programs
$$
\max \Big\{\langle \zeta^i,y\rangle \mid y\in \R^{\cF}_+, \sum_{F\in \cF}y_F\1_F\leq f\Big\} \quad(1\leq i\leq m).
$$
Let $\ov{x}^i$ denote optimal solutions for the dual linear programs
$$
\min \{\langle f,x\rangle \mid x\in \R^N_+, x(F)\geq \zeta^i(F),\;\forall F\in \cF\}.
$$

Consider the belief function $v=\sum_{i=1}^m\beta_i \zeta^i\in \cV^+$ with $\beta_i\geq 0$ and set $\ov{x}:=\sum_{i=1}^m \beta_i\ov{x}^i$. Then $\ov{x}$ is feasible for the linear program
$$
\min \{\langle f,x\rangle \mid x\in \R^N_+, x(F)\geq v(F)\;\forall F\in \cF\}
$$ and, in view of $\langle f,\ov{x}\rangle  = \sum_i\beta_i\langle
f,\ov{x}^i\rangle = \sum_i\beta_i\langle\zeta^i,\ov{y}\rangle = \langle v,\ov{y}\rangle $, optimal by linear programming
duality. So we find
$$
    \langle v,\ov{y}\rangle  = \int_\cF f \, dv \quad\mbox{for all belief functions $v\in \cV^+$}
$$
and consequently
\begin{eqnarray*}
 \langle v,\ov{y}\rangle &=&  \langle v^+,\ov{y}\rangle -\langle v^-,\ov{y}\rangle \\
 &=&\int f \, dv^+ - \int_\cF f \, dv^- \;=\;\int_\cF f \, dv
\end{eqnarray*}
for all valuations $\in \cV$.
\end{proof}

 The linearity of the Monge functional $v\mapsto [f](v)$ furnishes a sufficient condition for  the Choquet functional $v\mapsto \int_\cF f \,dv$ to be linear on $\cV$ (and thus to strengthen Proposition~\ref{p.subadditivity}):

\begin{corollary}\label{c.Choquet-representation} Assume that the Monge algorithm computes the Choquet integral for all simple functions $\zeta^i$. Then we have
$$
\int_\cF f \, dv = \sum_{i=1}^m \beta_i\int_\cF f \,d\zeta^i \quad\mbox{for all $v=\D\sum_{i=1}^m\beta_i \zeta^i \in \cV$.}
$$
\end{corollary}

\section{Ordering by containment}
We investigate in this section systems under the set-theoretic containment order relation $\subseteq$ and consider the system $(\cF,\subseteq)$. A fundamental observation is a simple expression for the integral relative to simple functions (which is well-known for the classical Choquet integral):

\begin{lemma}\label{lem:min}
Let $(\cF,\subseteq)$ be arbitrary and $f:N\rightarrow \R_+$. Then for any
$F\in \cF$,
\[
\int f \,d\zeta^F = \min_{j\in F}f_j.
\]
\end{lemma}
\begin{proof}
Let $s\in F$ be such that $f_s =\min_{j\in F} f_j$  and denote by $x^s\in \R^N_+$ the corresponding unit vector. Then $x^s$ is feasible for the linear program
$$
\min_{x\geq 0}~  \langle f,x\rangle  \quad\mbox{s.t.}\quad x(F') \geq 1\quad \mbox{for all $F'\in \cF$ with $F'\supseteq F$}
$$
while the vector $y^s\in \R^{\cF}_+$ with the only nonzero component $y^s_{F} =
f_s$ is feasible for the dual linear program
$$
\max_{y\geq 0}~ \langle \zeta^F,y\rangle   \quad\mbox{s.t.}\quad \sum_{F'\supseteq F}y_{F'}\1_{F'} \leq f.
$$ In view of $\langle f,x^s) = f_s = \langle \zeta^F, y^s\rangle $, linear programming duality
guarantees optimality and we conclude
\begin{equation}\label{eq.containment-simple-Choquet}
      \int f\,d \zeta^F \;=\; \min_{j\in F} f_j \quad\mbox{for all $F\in \cF$.}
\end{equation}
\end{proof}

\subsection{Extensions of valuations} A simple function $\zeta^F:\cF\to \{0,1\}$ (relative to $(\cF,\subseteq)$) corresponds naturally to a simple function $\hat{\zeta}^F:2^N\to \{0,1\}$ (relative to the Boolean algebra $(2^N,\subseteq)$, where
$$ \hat{\zeta}^F(S) = 1 \quad:\Longleftrightarrow\quad S\supseteq F
\quad\mbox{for all $S\subseteq N$.}
$$
We thus associate with any valuation $v=\sum_{F\in \cF}\beta_F\zeta^F\in \cV(\cF)$ its \emph{extension} $\hat{v}:2^N\to \R$, where
$$
   \hat{v}(S) := \sum_{F\in \cF} \beta_F \hat{\zeta}^F(S) \quad (S\subseteq N)
$$
and immediately observe $\hat{v}(F) = v(F)$ for all $F\in \cF$. If $\beta_F\geq 0$,
the function $\hat{v}$ is easily seen to be supermodular on $(2^N,\subseteq)$. By Proposition~\ref{prop:mono}, the classical Choquet integral operator $f\mapsto \int f\, d\hat{v}$  therefore yields an upper integral relative to $(\cF,\subseteq)$. Hence we have
$$
\int_\cF f\,dv \leq \int  f\, d\hat{v} \quad\mbox{for all belief functions $v\in \cV^+(\cF)$.} $$

We now present a class of systems where actually equality holds for the two integrals (see Corollary~\ref{c.Choquet-Choquet} below).

\subsubsection{Weakly union-closed systems}\label{sec.w-u-closed}
Assume that $\cF$ is \emph{weakly union-closed} in the sense
$$
 F\cap G\neq \emptyset \quad \Longrightarrow \quad F\cup G\in \cF \quad\mbox{for all $F,G\in \cF$}
$$
and consider the containment order $(\cF,\subseteq)$ as before. We assume $\cF=\{F_1,\ldots,F_m\}$ such that for all indices $1\leq i,j\leq m$,
$$
    F_i \subseteq F_j \quad \Longrightarrow \quad i\leq j.
$$
\begin{remark}
Weakly union-closed systems have been investigated by Algaba {\it et al.}
\cite{albibolo01} as \emph{union-stable systems} with respect to games on
communication graphs, where it is noted that a set system $\cF$ is weakly union-closed if and only if for every $G\subseteq N$, the
  maximal sets in $\cF(G)=\{F\in \cF\mid F\subseteq G\}$ are pairwise disjoint.
\end{remark}

\begin{lemma}\label{lem:wuc} Let $\cF$ be weakly union-closed and denote by $\cM=\{M_1,\ldots,M_q\}$ the sequence of subsets chosen in
  MA. Then $(\cM,\subseteq)$ forms a forest (\emph{i.e.}, a cycle-free
  subgraph of the Hasse diagram of $(\cF,\subseteq)$) in which all
  descendants of a node are pairwise disjoint. Moreover, the outputs $y$ and
  $\pi=p_1\cdots p_q$ of the algorithm yield
\[
[f](v) = \langle v,y\rangle  = \sum_{i=1}^q (f_{p_i}-f_{p_{\uparrow i}})v(M_i)
\]
where $\uparrow\! i$ is the index of father of node $M_i$ in the tree, and
$f_{p_{\uparrow i}}=0$ if it has no father.
\end{lemma}
\begin{proof} At iteration $i$, either $M_i\subseteq
M_{i-1}$ or $M_i\cap M_{i-1}=\emptyset$ holds since $\cF$ is weakly union-closed.  So $M_i$ cannot have two fathers (supersets).  Hence $\{M_1,\ldots,M_m\}$ is a tree. Descendants of a node are pairwise disjoint for the same reason. Now, the formula results from the updating rule of $c$.
\end{proof}
Note that $(\cM,\subseteq)$ becomes a tree (connected and cycle-free) if $N\in\cF$.

\begin{theorem}\label{t.Monge-w-u-c}
Let $(\cF,\subseteq)$ be a weakly union-closed system.
For all $f\in\R_+^n,v\in \cV$ with $v=\sum_{F\in \cF}\beta_F\zeta^F$, we have
\begin{equation}\label{eq.classical-Choquet}
 [f](v)=\int_\cF f \,dv = \sum_{F\in \cF}\beta_F\int f\,d\zeta^F=\sum_{F\in
   \cF}\beta_F\min_{i\in F}f_i.
\end{equation}
\end{theorem}

\begin{proof} Assume that the Monge algorithm outputs the vector $y\in \R^{\cF}_+$, the set family $\cM=\{M_1,\ldots,M_k\}$ and the sequence $\pi= p_1\ldots p_k$ upon the input $f\in \R^N_+$. Consider the simple function $\zeta^i$ and recall that $\zeta^i(M_j) = 1$ is equivalent with $F_i\subseteq M_j$.

\medskip
Let $s$ be the smallest index such that $p_s\in F_i$. Then we have $F_i\subseteq M_s$ (since $\cF$ is weakly union-closed and MA always selects a $\subseteq$-maximal member $M$ of the current system $\cF(X)$). So we find
$$
 \(\zeta^i, y\) =  \sum_{M_j \ni p_s} y_{M_j} = f_s = \min_{t\in F_i} f_t =\int f d\zeta^i.
$$ In view of Lemma~\ref{lem:min}, MA thus computes the Choquet integral
 correctly for simple functions. So Theorem~\ref{t.Monge} guarantees the claim
 of the Theorem to be true, and Lemma~\ref{lem:min} explains the last equality
 in (\ref{eq.classical-Choquet}).

\end{proof}

\begin{corollary}\label{c.Choquet-Choquet}  Let $(\cF,\subseteq)$ be weakly union-closed and $f\in \R_+^N$. Then
$$
   \int_\cF f\,dv = \int f\, d\hat{v}  \quad\mbox{holds for all valuations $v\in \cV(\cF)$},
$$
and $\hat{v}$ is determined by
\[
\hat{v}(S)  =\int_\cF \1_S\, dv = \sum_{F\text{ maximal in }\cF(S)}v(F),\quad \forall S\in 2^N.
\]

\end{corollary}

\begin{proof} Assume $v=\sum_{F\in \cF}\beta_F\zeta^F$. Then Corollary~\ref{c.Choquet-representation} yields
$$
 \int_\cF f\,dv =\sum_{F\in \cF} \beta_F \min_{f\in F} f_j = \int f\, d\hat{v}.
$$
The expression for $\hat{v}$ results from the Monge algorithm.
\end{proof}

\begin{remark}
\begin{enumerate}
\item Corollary \ref{c.Choquet-Choquet} shows that the Choquet integral
  on a weakly union-closed family essentially equals the classical Choquet
  integral, and therefore inherits all its properties (in particular, comonotonic
  additivity (see Proposition~\ref{prop:schm1})).
\item The fact that $\hat{v}$ is an extension of $v$ suggests the following
  interpretation: consider again the classical definition of the Choquet
  integral given by (\ref{eq:ccho}), but with $v$ defined on $\cF$ instead of
  $2^N$. Call $f\in\R_+^n$ \emph{$\cF$-measurable} if all level sets $\{i\in
  N\mid f_i\geq\alpha\}$ belong to $\cF$, and denote by $M(\cF)$ the set of all
  $\cF$-measurable nonnegative functions. Then the classical Choquet integral on
  $\cF$ coincides with the (general) Choquet integral for all measurable
  functions, and therefore the latter is an extension of the former from
  $M(\cF)$ to $\R_+^n$.
\item The extension $\hat{v}$ is well-known in cooperative game theory as
   Myerson's \cite{mye77a} \emph{restricted game} and is used in the analysis of communication graph games. There, $\cF$ is the collection of connected components of the graph with the property of being weakly union-closed arising naturally.
\item A capacity $v$ on $(\cF,\subseteq)$ may not yield
  $\hat{v}$ as a capacity on $(2^N,\subseteq)$. Consider for example $N=\{1,2,3,4,5\}$ and the weakly union-closed system
  $\cF=\{12345,1234,2345,1345,124,234,$ $345,12,35,2,5\}$. Then $\hat{v}(N)=v(N)$ and $\hat{v}(1235)=v(12)+v(35)$. $v(N)=1=v(12)=v(35)$ shows that $\hat{v}$ is not monotone. Therefore, the Choquet integral w.r.t. a capacity is not necessarily monotone in general.
\end{enumerate}
\end{remark}

From Proposition~\ref{prop:mono}, we immediately see:

\begin{corollary}\label{c.submodularity} Let $(\cF,\subseteq)$ be weakly union-closed and $v\in \cV$ an arbitrary valuation. Then the following are equivalent:
\begin{itemize}
\item[(i)] The operator $f\mapsto \int_\cF f\, dv$ is superadditive on $\R_+^N$.
\item[(ii)] The extension $\hat{v}:2^N\to \R$ of $v$ is supermodular.
\end{itemize}
\end{corollary}

\subsubsection{Algebras}
An \emph{algebra} is a collection $\cA$ of subsets of $N$ that is closed under set union and set complementation with $\emptyset,N\in \cA$. In particular, $\cF=\cA\setminus\{\emptyset\}$ is a weakly union-closed family. Let $\cB=\cB(\cA) = \{B_1,\ldots,B_k\}$ be the family of \emph{atoms} (\emph{i.e.}, minimal non-empty members) of the the algebra $\cA$. Then $(\cA,\subseteq)$ is isomorphic to $2^\cB$ (and, in particular, also intersection-closed).

\medskip
Lehrer~\cite{leh09a} (see also Teper~\cite{tep09}) has introduced a discrete integral relative to the algebra $\cA$ as follows. Given a probability measure $P$ on $\mathcal{A}$ and a non-negative function $f\in \R_+^N$, define
\[
\int_\cL f\,dP_{\mathcal{A}}:= \sup_{\lambda\geq 0}\Big\{\sum_{S\in\mathcal{A}}\lambda_SP(S)\mid
\sum_{S\in\mathcal{A}}\lambda_S\1_S\leq f\Big\}.
\]
Lehrer shows that the functional $f\mapsto \int_\cL f\, dP_\cA$  is a concave operator on $\R_+^N$. Let us exhibit Lehrer's integral as a special case of our general Choquet integral.

\begin{proposition}\label{l.Lehrer-integral} Let $\cA$ be an algebra and $P$ a probability measure on $\cA$. Setting $\cF=\cA\setminus\{\emptyset\}$, one then has
$$
\int_\cL f\,dP_{\mathcal{A}} = \int_\cF f\,dP \quad\mbox{for all $f\in \R^N_+$.}
$$
In particular, Lehrer's integral can be computed with the Monge algorithm.
\end{proposition}

\begin{proof} Because of $P(\emptyset) = 0$, we have
$$
\int_\cL f\,dP_{\mathcal{A}} = \max_{y\geq 0}\Big\{\sum_{S\in\mathcal{F}} y_SP(S)\mid
\sum_{S\in\mathcal{F}}y_S\1_S\leq f\Big\}.
$$
 By Lemma~\ref{l.Choquet-upper-integral}, the Proposition now follows once we establish $P$ as a belief function relative to $(\cF,\subseteq)$. Indeed, as a probability measure $P$ is additive on $\mathcal{A}$, we infer the M\"obius representation
$$
     P = \sum_{B\in \mathcal{B}(\mathcal{A})} P(B)\zeta^{B}
$$
with non-negative coefficients $P(B)\geq 0$, which proves the Proposition.

\end{proof}

\medskip
Lehrer furthermore defines the induced capacity $v_{\mathcal{A}}$ on $2^N$ by
\[
v_{\mathcal{A}}(S) := \sup\{P(A)\mid A\in\mathcal{A}, A\subseteq S\}.
\]

\begin{lemma} Let $P$ be a probability measure on the algebra $\cA$. Then the induced capacity $v_\cA$ is precisely the extension of $P$, \emph{i.e.},
$$
    v_{\cA}(S) = \hat{P}(S) \quad \mbox{holds for all $S\subseteq N$.}
$$
\end{lemma}

\begin{proof} Let $\cB(S) = \{B\in \cB\mid B\subseteq S\}$ be the collection of all atoms that are contained in $S$. Since $P$ is non-negative and additive on $\cA$, we apparently have
$$
\sup\{P(A)\mid A\in\mathcal{A}, A\subseteq S\} = \sum_{B\in \cB(S)}P(B) =\sum_{B\in \cB} P(B)\hat{\zeta}^B(S) =\hat{P}(S).
$$
\end{proof}

\section{Intersection systems}
We address in this section a more general order relation than the containment
order. It has applications in graph theory (namely, the cut set problem, see
\cite{joh65}) we do not detail here since this falls outside the scope of the
paper. This order relation will permit to derive general results on supermodularity.

\subsection{Consecutive ordered systems}
The (partial) precedence ordering $(\cF,\preceq)$ is said to be
\emph{consecutive} if
$$
F\cap H\;\subseteq\; G \quad\mbox{holds for all $F,G,H\in \cF$ with $F\preceq G\preceq H$.}
$$
The consecutive property implies a kind of submodularity condition: For any $F,G\in \cF$ and $L,U\in \cF(F\cup G)$ with $L\preceq F,G\preceq L$, we find
\begin{equation}\label{eq.submodular}
\1_L + \1_U\leq \1_{F} + \1_{G}.
\end{equation}
(The familiar form of the submodular inequality appears intuitively in (\ref{eq.submodular}) when we employ the notation $F\wedge G := L$ and $F\vee G:= U$.)

\medskip
We call a consecutive ordered system $(\cF,\preceq)$ an \emph{intersection
  system} if for all sets $F\in \cF$ the following is true:
\begin{itemize}
\item[(IS0)] For every $G\in \cF$ with  $F\cap G\neq \emptyset$, there is some $F\vee G\in \cF(F\cup G)$ such that $F,G\preceq F\vee G$.
\item[(IS1)] The \emph{upper interval} $[F) :=\{G\in \cF\mid G\succeq F\}$ is
  ``closed'' under $\vee,\wedge$, \emph{i.e.}, for every $G,H\in[F)$ there exist sets $G
    \vee H,G\wedge H\in\cF(G\cup H)$ such that $$ F\preceq G\wedge H\preceq G,H\preceq
    G\vee H.$$
\end{itemize}

\begin{remark}
Note that every containment order $(\cF,\subseteq)$ is trivially consecutive. So every weakly union-closed family $\cF$ yields $(\cF,\subseteq)$ as an intersection system: In the case $G,H\in [F)$, it would suffice to take
$$
G\wedge H:= F \quad\mbox{and}\quad G\vee H:= G\cup H.
$$
So intersection systems generalize the classical model $(2^N,\subseteq)$ in particular.
\end{remark}

As an illustration, we give an example of intersecting system where the order is
not the containment order.
\begin{example}\label{ex:1}
  Let $N=\{1,2,3,4,5,6\}$ and consider the system below. It can be checked that
  it is an intersection system.
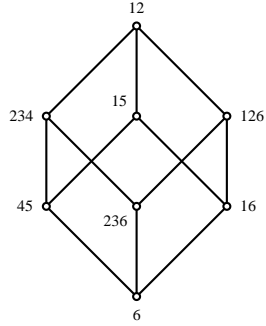
\begin{figure}[htb]
\begin{center}
\psset{unit=0.6cm}
\pspicture(-0.5,-1.5)(4.5,5.5)
\pspolygon(0,1)(0,3)(2,5)(4,3)(4,1)(2,-1)
\psline(0,1)(2,3)(4,1)
\psline(0,3)(2,1)(4,3)
\psline(2,-1)(2,1)
\psline(2,3)(2,5)
\pscircle[fillstyle=solid](0,1){0.1}
\pscircle[fillstyle=solid](0,3){0.1}
\pscircle[fillstyle=solid](2,5){0.1}
\pscircle[fillstyle=solid](4,3){0.1}
\pscircle[fillstyle=solid](4,1){0.1}
\pscircle[fillstyle=solid](2,-1){0.1}
\pscircle[fillstyle=solid](2,1){0.1}
\pscircle[fillstyle=solid](2,3){0.1}
\uput[180](0,1){\tiny 45}
\uput[180](0,3){\tiny 234}
\uput[90](2,5){\tiny 12}
\uput[0](4,3){\tiny 126}
\uput[0](4,1){\tiny 16}
\uput[-90](2,-1){\tiny 6}
\uput[225](2,1){\tiny 236}
\uput[135](2,3){\tiny 15}
\endpspicture
\end{center}
\end{figure}
\end{example}

Our main result in this section assures that the Choquet integral on intersection
systems may be computed with the Monge algorithm.

\begin{theorem}\label{thm:inter}
Let $(\cF,\preceq)$ be an intersection system, $f\in \R_+^N$ and $(\ov{y},\cM,\pi)$
the corresponding output of the Monge algorithm. Then we have
\[
[f](v)=\int_{\cF} f \,dv = \sum_{F\in \cF}\beta_F\int_{\cF}f\, d\zeta^F \quad\mbox{for all valuations $v\in \cV$.}
\]
\end{theorem}

\begin{proof}
By Theorem~\ref{t.Monge}, it suffices to establish the Theorem for simple
functions. Choose   $v=\zeta^i$, the simple function associated to $F_i$. Assume $\cM=\{M_1,\ldots,M_k\}$ and $\pi=\{p_1,\ldots, p_k\}$ and set
$$
     \cS := \{M_1\}\cup [F_i) = \{M_1\}\cup \{F\in \cF\mid \zeta^i(F)\neq 0\}.
$$
Let $y^*\in \R_+^\cS$ be the (with respect to the index order of $\cF$) lexicographically maximal vector with the property
\begin{equation}\label{eq:feas}
 \sum_{S\in\cS}y^*_S \1_S\leq f\quad\mbox{and}\quad \sum_{F\in \cS} y_S^* \zeta^i(S) = \int_\cF f\,d\zeta^i.
\end{equation}
It suffices to show $f_{p_1}=\ov{y}_{M_1}=y^*_{M_1}$. (The Theorem then follows by induction on $|N|$ because $\cF(N\setminus\{p_1\})$ is also an intersection system.)

\medskip

Since $\ov{y}_{M_1}=\min_{M_1}f$, the selection rule of MA guarantees $0\leq
y^*_{M_1}\leq \ov{y}_{M_1}$. Suppose $\ov{y}_{M_1}>y^*_{M_1}$ were the
case. Then there must exist some $S\in \cS\setminus\{M_1\}$ with $y_S^*>0$ and
$M_1\cap S\neq \emptyset$ (because otherwise $y^*_{M_1}$ could be increased
without violating the feasibility conditions, which contradicts the
lexicographic maximality of $y^*$).
Set
\[\cC:=\{S\in \cS\mid y^*_S >0, M_1\cap S\neq \emptyset\}
\]
the collection of such sets. For any $S\in\cC$, by (IS0) $S\vee M_1$ exists and
$S\vee M_1 \succeq M_1$. By the selection rule of MA, we conclude
$$
M_1= S\vee M_1 \succeq S\succeq F_i \quad\mbox{and hence}\quad M_1\in [F_i).
$$
Moreover, $M_1$ is the unique maximal member of $[F_i) =\cS$.

\medskip
\textsc{Claim:} $\cC$ is a chain in $(\cF,\preceq)$.

\medskip
Indeed, if there existed incomparable sets $F,G\in \cC$, then $\cC\subseteq [F_i)$ implies the existence of $F\wedge G$ and $F\vee G$ in $[F_i)$ by property (IS1). So we could decrease
$y^*$ on the sets $F$ and $G$ by $\varepsilon:\leq\min\{y^*_F,y^*_G\}>0$ and
increase it on $F\wedge G$ and $F\vee G$ by the same value $\varepsilon$.
The resulting vector
$$
y'=y^*
  +\varepsilon(\1_{F\vee G} + \1_{F\wedge G} - \1_F - \1_G)
$$
would be lexicographically larger than $y^*$ and satisfy the right equality in
  (\ref{eq:feas}). Moreover, from (\ref{eq.submodular}) we deduce that $y'$
  still satisfies the left inequality in (\ref{eq:feas}) and thus contradict the
  choice of $y^*$.

\medskip
Let $C$ be the maximal element of $\cC\setminus\{M_1\}$ and observe from the consecutiveness of $(\cF,\preceq)$:
$$
    C \supseteq M_1\cap S \quad\mbox{for all $S\in \cS\setminus\{M_1\}$.}
$$
So $y^*$ could be increased on $M_1$ by $y_{C}^*>0$ and decreased on  $C$ by the same amount without violating feasibility -- again in contradiction to the choice of $y^*$. We therefore find  $y^*_{M_1} =\ov{y}_{M_1}$, which establishes the Theorem.
\end{proof}

\subsection{Supermodularity and superadditivity}\label{sec:union-closed-systems}
In view of Theorem~\ref{thm:inter}, the Choquet functional $f\mapsto \int_\cF f \,dv$ is superadditive on $\R_+^N$ if $v$ is a belief function on the intersection system $(\cF,\preceq)$. Unfortunately, no analogue of Proposition~\ref{prop:mono} is known for general ordered systems that would provide a ''combinatorial'' characterization of valuations $v$ with superadditive Choquet integral. We will now exhibit a model that generalizes classical supermodular functions and is sufficient for superadditivity.

\medskip
Let $(\cF,\preceq)$ be a consecutive ordered system and $\cF=\{F_1,\ldots,F_m\}$ arranged so that for all $1\leq i,j\leq m$,
$$
      F_i \succeq F_j \quad \Longrightarrow \quad i\leq j.
$$
We call two sets $F_i,F_j\in \cF$ \emph{co-intersecting} if there exists some index $k\leq \min\{i,j\}$ such that $F_k\cap F_i\neq \emptyset$ and $F_k\cap F_j\neq \emptyset$. ($F_i\cap F_j=\emptyset$ may be permitted for co-intersecting sets $F_i,F_j$).

\medskip
It is convenient to augment the ordered system $(\cF,\preceq)$ to the order $(\cF_0,\preceq)$, where $\cF_{0}:=\cF\cup\{\emptyset\}$ and $\emptyset$ is the unique minimal element.
Morever, we extend any valuation $v$ to a function on $\cF_0$ {\it via} the normal property $v(\emptyset):=0$. We now say that a valuation $v$ is \emph{supermodular} if for all co-intersecting sets $F,G\in \cF$, there are sets $F\wedge G,F\vee G\in \cF_0(F\cup G)$ such that
\begin{itemize}
\item[(S1)] $F\wedge G\preceq F,G\preceq F\vee G$.
\item[(S2)] $v(F\wedge G)+v(F\vee G) \geq v(F) +v(G)$.
\end{itemize}

Recalling that a capacity is a non-negative and monotone valuation, we can show with the technique of Theorem~\ref{thm:inter}:

\medskip
\begin{theorem}\label{t.supermodularity} Let $v:\cF\to \R_+$ be a supermodular capacity on the intersection system $(\cF,\preceq)$. Then
$$
   \int_\cF f \,dv = \max_{y\geq 0} \Big\{\(v,y\)\mid \sum_{F\in \cF} y_F\1_F \leq f\Big\} \
\quad\mbox{holds for all $f\in \R_+^N$.}
$$
Hence $f\mapsto \int_\cF f\,dv$ is a positively homogeneous and superadditive functional.
\end{theorem}

\begin{proof} Since $(\cF,\preceq)$ is an intersection system, the Monge algorithm computes the Choquet integral. Consequently, it suffices to prove
$$
\tilde{v}(f):= \max_{y\geq 0} \Big\{\(v,y\)\mid \sum_{F\in \cF} y_F\1_F \leq f\Big\} = \sum_{F\in \cF}\ov{y}_Fv(F),
$$
 where we assume $(\ov{y},\cM,\pi)$ to be the output of MA with
 $\cM=\{M_1,\ldots,M_k\}$ and $\pi=\{p_1,\ldots, p_k\}$ (\emph{cf.} the proof of Lemma~\ref{l.Choquet-upper-integral}).  Set
$$
     \cS := \{M_1\}\cup \{F\in \cF\mid v(F)\neq 0\}
$$
and let $y^*\in \R_+^\cS$ be the (with respect to the index order of $\cF$)
lexicographically maximal vector with the optimality property
$$
 \sum_{S\in\cS}y^*_S \1_S\leq f\quad\mbox{and}\quad \sum_{S\in\cS}y^*_Sv(S) = \tilde{v}(f).
$$
We will argue that the assumption $\ov{y}_{M_1} > y^*_{M_1}$ would lead to a contradiction.

\medskip
Indeed, there must exist some $S\in \cS\setminus\{M_1\}$ with $y_S^*>0$ and $M_1\cap S\neq \emptyset$. So we have $S\vee M_1 \succeq M_1$ and hence
$M_1= S\vee M_1 \succeq S$, \emph{i.e.}, $M_1$ is the unique maximal member of
$$
\cC:=\{S\in \cS\mid y^*_S >0, M_1\cap S\neq \emptyset\}.
$$
 So any two members $F,G\in\cC$ are co-intersecting.

\medskip
\textsc{Claim:} $\cC$ is a chain in $(\cF,\preceq)$.

\medskip
 Suppose $\cC$ did contain incomparable sets $F,G$. Then we could decrease $y^*$ on the sets $F$ and $G$ by $\varepsilon:=\min\{y^*_F,y^*_G\}>0$ and increase it on $F\wedge G$ and $F\vee G$ by the same value $\varepsilon>0$. Let $y'$ be the resulting vector. Because of the supermodularity of $v$, we have
 \begin{multline}
 v(F\wedge G) y'_{F\wedge G} +  v(F\vee G) y'_{F\vee G} + v(F)y'_{F} + v(G) y'_{G}
\\ \geq\; v(F\wedge G) y^*_{F\wedge G} +  v(F\vee G) y^*_{F\vee G}+ v(F) y^*_{F} +v(G) y^*_{G}.\label{eq:super}
\end{multline}
In view of $F\vee G \succeq F$ and the monotonicity of the capacity $v$, we have
$$
v(F\vee G) \geq v(F) >0\quad\mbox{and thus}\quad F\vee G\in \cS.
$$
So $y'$ would be lexicographically larger than $y^*$,  feasible by (\ref{eq.submodular}) and optimal by (\ref{eq:super}), which contradicts the choice of $y^*$.

\medskip
Let $C$ be the largest member of $\cC\setminus\{M_1\}$ and observe from the consecutiveness of $(\cF,\preceq)$:
$$
    C \supseteq M_1\cap S \quad\mbox{for all $S\in \cS\setminus\{M_1\}$.}
$$
So $y^*$ could be increased on $M_1$ by $y_{C}^*>0$ and decreased on  $C$ by the same amount without violating feasibility -- again in contradiction to the choice of $y^*$. We therefore  conclude $y^*_{M_1} =\ov{y}_{M_1}$, which establishes the Theorem.
\end{proof}

\bigskip
\subsubsection{Union-closed systems}
 As an application of Theorem~\ref{t.supermodularity}, consider a family $\cF$
 that is closed under taking arbitrary unions. Then $(\cF,\subseteq)$ is an
 intersection system in particular and $\cF_0$ is closed under the well-defined operations
\begin{align*}
F'\vee F''& := F'\cup F''\\
F'\wedge F''& := \bigcup\{F\in \cF_0\mid F\subseteq F'\cap F''\}.
\end{align*}

Theorem~\ref{t.supermodularity} allows us to establish the following generalization of Lov\'asz'~\cite{lov83} result (Proposition~\ref{prop:mono}).

\begin{theorem}\label{th:vhatv}
Assume that $\cF$ is union-closed and $v$ a capacity on $(\cF,\subseteq)$. Then the following statements are equivalent:
\begin{enumerate}
\item $\D\int_\cF f\,dv = \max\Big\{\(v,y\)\mid y \in\R^\cF_+,\sum_{F \in\cF}y_F\1_F\leq f\Big\}$ for all $f\in\mathbb{R}_+^n$.
\item The functional $f\mapsto\int_\cF f\, dv$ is superadditive on $\R^N_+$.
\item $v$ is supermodular.
\end{enumerate}
\end{theorem}
\begin{proof}
(iii) $\Rightarrow$ (i) $\Rightarrow$ (ii) follows from Theorem~\ref{t.supermodularity}. We show (ii) $\Rightarrow$ (iii):
\[
\int_\cF \frac{1}{2}\1_F\,dv + \int_\cF\frac{1}{2}\1_{F'}\,dv\leq\int_\cF\frac{1}{2}(\1_F+\1_{F'})\,dv
\]
yields
\[
\frac{1}{2}v(F) + \frac{1}{2}v(F') \leq \frac{1}{2}(v(F\cup F') + v(F\wedge F'))
\]
for any $F,F'\in \cF$. Hence $v$ is supermodular.
 \end{proof}

\begin{corollary}\label{cor:8}
Let $\cF$ be a union-closed and $v$ a capacity with extension $\hat{v}$ on $(\cF,\subseteq)$. Then the following statements are equivalent:
\begin{enumerate}
\item $v:\cF\to \R$ is supermodular on $(\cF,\subseteq)$.
\item $\hat{v}:2^N\to \R$ is supermodular on $(2^N,\subseteq)$.
\end{enumerate}
\end{corollary}

\section*{Acknowledgment}
The second author is indebted to Toshiaki Murofushi for fruitful discussions on
an earlier version of this paper.

\bibliographystyle{plain}
\bibliography{../BIB/fuzzy,../BIB/grabisch,../BIB/general}

\end{document}